\makeatletter
\def\input@path{{styles/}{../styles/}}
\makeatother

\makeatletter
\def\input@path{{styles/}{../styles/}}
\makeatother

\ifx\SODA\undefined%
    \ifx\DCG\undefined%
        \documentclass[12pt]{article}%
        \providecommand{\SODAVer}[1]{}%
        \providecommand{\RegVer}[1]{#1}%
        \providecommand{\DCGVer}[1]{}%
        \def\UseBibLatex{1}
    \else
        \documentclass[sn-mathphys]{sn-jnl}%
        \providecommand{\SODAVer}[1]{}%
        \providecommand{\RegVer}[1]{}%
        \providecommand{\DCGVer}[1]{#1}%
    \fi
\else%
    \documentclass[twoside,leqno]{article}
    \providecommand{\SODAVer}[1]{#1}%
    \providecommand{\RegVer}[1]{#1}%
    \providecommand{\DCGVer}[1]{}%
    \usepackage{ltexpprt}%
    \newcommand\relatedversion{}
    \renewcommand\relatedversion{\thanks{The full version of the paper
          can be accessed at
          \protect\url{https://arxiv.org/abs/2112.14829}}}%
\fi
\RegVer{%
   \usepackage[cm]{fullpage}%
   \usepackage{caption}
}%

\DCGVer{%
   \usepackage{textcomp}%
}

\ifx\sarielComp\undefined%
\newcommand{\SarielComp}[1]{}
\newcommand{\NotSarielComp}[1]{#1}%
\else
\newcommand{\SarielComp}[1]{#1}%
\newcommand{\NotSarielComp}[1]{}%
\fi

\usepackage{amsmath}%
\usepackage{amssymb}%
\usepackage{xcolor}%
\usepackage{scalerel}%

\SarielComp{\usepackage{sariel_colors}}%

\RegVer{%
   \usepackage[amsmath,thmmarks]{ntheorem}%
   \theoremseparator{.}%

   \usepackage{titlesec}%
   \titlelabel{\thetitle. }%
}%
\usepackage{xcolor}%
\usepackage{mleftright}%
\usepackage{xspace}%
\usepackage{hyperref}%
\usepackage{mathcalb}%

\usepackage{stmaryrd}%
\providecommand{\IntRange}[1]{\mleft\llbracket #1 \mright\rrbracket}
\newcommand{\IRX}[1]{\IntRange{#1}}%
\newcommand{\IRY}[2]{\left\llbracket #1:#2 \right\rrbracket}

\hypersetup{%
      breaklinks,%
      colorlinks=true,%
      urlcolor=[rgb]{0.25,0.0,0.0},%
      linkcolor=[rgb]{0.5,0.0,0.0},%
      citecolor=[rgb]{0,0.2,0.445},%
      filecolor=[rgb]{0,0,0.4},
      anchorcolor=[rgb]={0.0,0.1,0.2}%
   }

\usepackage[ocgcolorlinks]{ocgx2}

\DCGVer{%
   \theoremstyle{thmstyleone}%
   \newtheorem{theorem}{Theorem}%
   \newtheorem{proposition}[theorem]{Proposition}%
   \newtheorem{lemma}[theorem]{Lemma}%
   \newtheorem{corollary}[theorem]{Corollary}%

   \theoremstyle{thmstyletwo}%
   \newtheorem{example}{Example}%
   \newtheorem{remark}{Remark}%
   \newtheorem{problem}{Problem}%
   \newtheorem{remarks}{Remarks}%
   \newtheorem{fact}{Fact}%

   \theoremstyle{thmstylethree}%
   \newtheorem{definition}{Definition}%

   \raggedbottom
}

\RegVer{%
\theoremseparator{.}%

\theoremstyle{plain}%
\newtheorem{theorem}{Theorem}[section]

\newtheorem{lemma}[theorem]{Lemma}

\newtheorem{corollary}[theorem]{Corollary}
\newtheorem{fact}[theorem]{Fact}

\theoremstyle{plain}%
\theoremheaderfont{\sf} \theorembodyfont{\upshape}%
\newtheorem*{remark:unnumbered}[theorem]{Remark}%
\newtheorem*{remarks}[theorem]{Remarks}%
\newtheorem{remark}[theorem]{Remark}%
\newtheorem{definition}[theorem]{Definition}

\newtheorem{problem}[theorem]{Problem}

\newcommand{\myqedsymbol}{\rule{2mm}{2mm}}

\theoremheaderfont{\em}%
\theorembodyfont{\upshape}%
\theoremstyle{nonumberplain}%
\theoremseparator{}%
\theoremsymbol{\myqedsymbol}%
\newtheorem{proof}{Proof:}%

}

\providecommand{\BibLatexMode}[1]{}
\providecommand{\BibTexMode}[1]{#1}

\ifx\UseBibLatex\undefined%
  \renewcommand{\BibLatexMode}[1]{}
  \renewcommand{\BibTexMode}[1]{#1}
\else
  \renewcommand{\BibLatexMode}[1]{#1}
  \renewcommand{\BibTexMode}[1]{}
\fi

\BibLatexMode{%
   \usepackage[bibencoding=utf8,style=alphabetic,backend=biber]{biblatex}%
   \usepackage{sariel_biblatex}%
}

\definecolor{blue25emph}{rgb}{0, 0, 11}
\providecommand{\emphic}[2]{%
   \textcolor{blue25emph}{%
      \textbf{\emph{#1}}}%
   \index{#2}}

\providecommand{\emphi}[1]{\emphic{#1}{#1}}

\definecolor{almostblack}{rgb}{0, 0, 0.3}
\providecommand{\emphw}[1]{{\textcolor{almostblack}{\emph{#1}}}}%

\newcommand{\atgen}{\symbol{'100}}
\newcommand{\SarielThanks}[1]{\thanks{Department of Computer Science;
      University of Illinois; 201 N. Goodwin Avenue; Urbana, IL,
      61801, USA; {\tt sariel\atgen{}illinois.edu}; {\tt
         \url{http://sarielhp.org/}.} #1}}

\newcommand{\HLink}[2]{\hyperref[#2]{#1~\ref*{#2}}}
\newcommand{\HLinks}[2]{\hyperref[#2]{#1\ref*{#2}}}
\newcommand{\HLinkSuffix}[3]{\hyperref[#2]{#1\ref*{#2}{#3}}}

\newcommand{\thmlab}[1]{{\label{theo:#1}}}
\newcommand{\thmref}[1]{\HLink{Theorem}{theo:#1}}
\newcommand{\thmrefs}[1]{\HLinks{T}{theo:#1}}

\newcommand{\corlab}[1]{\label{cor:#1}}
\newcommand{\corref}[1]{\HLink{Corollary}{cor:#1}}%

\providecommand{\deflab}[1]{\label{def:#1}}
\newcommand{\defref}[1]{\HLink{Definition}{def:#1}}
\newcommand{\defrefY}[2]{\hyperref[def:#2]{#1}}

\newcommand{\Dual}[1]{#1^\star}

\newcommand{\factlab}[1]{\label{fact:#1}}%
\newcommand{\factref}[1]{\HLink{Fact}{fact:#1}}%

\providecommand{\TPDF}[2]{\texorpdfstring{#1}{#2}}

\newcommand{\levelX}[1]{\Mh{\mathrm{level}}\pth{#1}}%
\newcommand{\tbllab}[1]{\label{table:#1}}
\newcommand{\tblref}[1]{\HLink{Table}{table:#1}}

\newcommand{\seclab}[1]{\label{sec:#1}}
\newcommand{\secref}[1]{\HLink{Section}{sec:#1}}

\newcommand{\apndlab}[1]{\label{apnd:#1}}
\newcommand{\apndref}[1]{\HLink{Appendix}{apnd:#1}}

\newcommand{\lemlab}[1]{\label{lemma:#1}}
\newcommand{\lemref}[1]{\HLink{Lemma}{lemma:#1}}%

\newcommand{\remlab}[1]{\label{rem:#1}}
\newcommand{\remref}[1]{\HLink{Remark}{rem:#1}}%

\providecommand{\eqlab}[1]{}%
\renewcommand{\eqlab}[1]{\label{equation:#1}}
\newcommand{\Eqref}[1]{\HLinkSuffix{Eq.~(}{equation:#1}{)}}

\newcommand{\SaveContent}[2]{%
   \expandafter\newcommand{#1}{#2}%
}

\newcommand{\RestatementOf}[2]{
   \noindent%
   \textbf{Restatement of #1.}
   {\em #2{}}%
}

\providecommand{\remove}[1]{}%
\newcommand{\Set}[2]{\left\{ #1 \;\middle\vert\; #2 \right\}}
\newcommand{\pth}[2][\!]{\mleft({#2}\mright)}%

\newcommand{\ceil}[1]{\left\lceil {#1} \right\rceil}
\newcommand{\floor}[1]{\left\lfloor {#1} \right\rfloor}

\newcommand{\cardin}[1]{\lvert {#1} \rvert}%

\renewcommand{\th}{th\xspace}

\renewcommand{\Re}{\mathbb{R}}%

\usepackage[inline]{enumitem}

\newlist{compactenumA}{enumerate}{5}%
\setlist[compactenumA]{topsep=0pt,itemsep=-1ex,partopsep=1ex,parsep=1ex,%
   label=(\Alph*)}%

\newlist{compactenuma}{enumerate}{5}%
\setlist[compactenuma]{topsep=0pt,itemsep=-1ex,partopsep=1ex,parsep=1ex,%
   label=(\alph*)}%

\newlist{compactenumI}{enumerate}{5}%
\setlist[compactenumI]{topsep=0pt,itemsep=-1ex,partopsep=1ex,parsep=1ex,%
   label=(\Roman*)}%

\newlist{compactenumi}{enumerate}{5}%
\setlist[compactenumi]{topsep=0pt,itemsep=-1ex,partopsep=1ex,parsep=1ex,%
   label=(\roman*)}%

\newlist{compactitem}{itemize}{5}%
\setlist[compactitem]{topsep=0pt,itemsep=-1ex,partopsep=1ex,parsep=1ex,%
   label=\ensuremath{\bullet}}%

\providecommand{\Mh}[1]{#1}%

\renewcommand{\P}{\PS}%
\newcommand{\PS}{\Mh{P}}%
\newcommand{\BS}{\Mh{\mathcal{O}}}%
\newcommand{\G}{\Mh{G}}%
\providecommand{\p}{\Mh{p}}%
\renewcommand{\p}{\Mh{p}}%

\newcommand{\etal}{\textit{et~al.}\xspace}

\newcommand{\D}{\Mh{\mathcal{D}}}%

\newcommand{\rect}{\Mh{\mathcalb{b}}}%
\newcommand{\obj}{\Mh{\mathcalb{o}}}%
\newcommand{\I}{\Mh{\mathcalb{l}}}%
\newcommand{\range}{\Mh{\nu}}%

\newcommand{\CHX}[1]{\mathsf{ch}\pth{#1}}%
\newcommand{\PPS}{\Mh{\mathcal{P}}}%
\newcommand{\canonX}[1]{\Mh{\mathsf{dy}\pth{#1}}}

\newcommand{\Family}{\Mh{\mathcal{F}}}
\newcommand{\FamilyB}{\Mh{\mathcal{G}}}%

\newcommand{\BY}[2]{\mathcal{L}_{#1}(#2)}
\newcommand{\permut}[1]{\left\langle {#1} \right\rangle}
\providecommand{\Matousek}{Matou{\v s}ek\xspace}
\newcommand{\Arr}{\mathop{\mathrm{\mathcal{A}}}}
\newcommand{\ArrX}[1]{\Arr\pth{#1}}%

\newcommand{\eps}{\varepsilon}

\newcommand{\VC}{\Term{VC}\xspace}%
\newcommand{\Term}[1]{\textsf{#1}}

\newcommand{\FZero}{\Mh{\mathcal{U}}}%

\newcommand{\Cutting}{\Mh{\Xi}}

\newcommand{\nF}{\Mh{m}}

\newcommand{\normX}[1]{\left\|#1\right\|}

\SODAVer{%
   \newtheorem{definition}{Definition}[section]%
   \newtheorem{remarks}{Remarks}[section]%
}

\BibLatexMode{%
   \bibliography{klan}%
}

\begin{document}
\DCGVer{%
   \title[On the Number of Incidences When Avoiding an Induced
   Biclique]{%
      On the Number of Incidences When Avoiding an Induced Biclique in
      Geometric Settings\DCGVer{$\phantom{}^{\text{\small{1}}}$}%
      \SODAVer{\relatedversion}%
      \RegVer{\footnote{A preliminary version of this paper appeared
            in SODA 2023 \cite{ch-oniwa-23}.}}%
   }%
} \RegVer{%
   \title{%
      On the Number of Incidences When Avoiding an Induced Biclique in
      Geometric Settings\DCGVer{$\phantom{}^{\text{\small{1}}}$}%
      \SODAVer{\relatedversion}%
      \RegVer{\footnote{A preliminary version of this paper appeared
            in SODA 2023 \cite{ch-oniwa-23}.}}%
   }%
}

\DCGVer{\footnotetext[1]{A preliminary version of this paper appeared
      in SODA 2023 \cite{ch-oniwa-23}.}}%

\DCGVer{}%

\DCGVer{%
   \author[1]{\fnm{Timothy M.} \sur{Chan}}\email{tmc@illinois.edu}

   \author[1]{\fnm{Sariel} \sur{Har-Peled}}\email{sariel@illinois.edu}

   \affil[1]{\orgdiv{Department of Computer Science}, %
      \orgname{University of Illinois, Urbana-Champaign}, %
      \orgaddress{\street{201 N. Goodwin Avenue}, \city{Urbana},
         \postcode{61801}, \state{IL}, \country{USA}}} }

\RegVer{%
   \author{Timothy M. Chan%
      \thanks{Department of Computer Science; University of Illinois;
         201 N. Goodwin Avenue; Urbana, IL, 61801, USA; {\tt
            tmc\atgen{}illinois.edu}.  Work on this paper was
         partially supported by an NSF AF award CCF-2224271.}%
      \and%
      Sariel Har-Peled%
      \SarielThanks{Work on this paper was partially supported by an
         NSF AF award CCF-1907400.}%
   }%
}

\date{\today}

\RegVer{\maketitle}%

\DCGVer{%
\abstract{%
   Given a set of points $\P$ and a set of regions $\BS$, an
   \emph{incidence} is a pair $(\p,\obj ) \in \P \times \BS$ such that
   $\p \in \obj$.  We obtain a number of new results on a classical
   question in combinatorial geometry: What is the number of
   incidences (under certain restrictive conditions)?

   We prove a bound of $O\bigl( k n(\log n/\log\log n)^{d-1} \bigr)$
   on the number of incidences between $n$ points and $n$
   axis-parallel boxes in $\Re^d$, if no $k$ boxes contain $k$ common
   points, that is, if the incidence graph between the points and the
   boxes does not contain $K_{k,k}$ as a subgraph.  This new bound
   improves over previous work, by Basit, Chernikov, Starchenko, Tao,
   and Tran (2021), by more than a factor of $\log^d n$ for $d >2$.
   Furthermore, it matches a lower bound implied by the work of
   Chazelle (1990), for $k=2$, thus settling the question for points
   and boxes.

   We also study several other variants of the problem. For
   halfspaces, using shallow cuttings, we get a linear bound in two
   and three dimensions.  We also present linear (or near linear)
   bounds for shapes with low union complexity, such as pseudodisks
   and fat triangles. %
}
}
\RegVer{%
   \begin{abstract}
       Given a set of points $\P$ and a set of regions $\BS$, an
       \emph{incidence} is a pair $(\p,\obj ) \in \P \times \BS$ such
       that $\p \in \obj$.  We obtain a number of new results on a
       classical question in combinatorial geometry: What is the
       number of incidences (under certain restrictive conditions)?

       We prove a bound of
       $O\bigl( k n(\log n/\log\log n)^{d-1} \bigr)$ on the number of
       incidences between $n$ points and $n$ axis-parallel boxes in
       $\Re^d$, if no $k$ boxes contain $k$ common points, that is, if
       the incidence graph between the points and the boxes does not
       contain $K_{k,k}$ as a subgraph.  This new bound improves over
       previous work, by Basit, Chernikov, Starchenko, Tao, and Tran
       (2021), by more than a factor of $\log^d n$ for $d >2$.
       Furthermore, it matches a lower bound implied by the work of
       Chazelle (1990), for $k=2$, thus settling the question for
       points and boxes.

       We also study several other variants of the problem. For
       halfspaces, using shallow cuttings, we get a linear bound in
       two and three dimensions.  We also present linear (or near
       linear) bounds for shapes with low union complexity, such as
       pseudodisks and fat triangles.
   \end{abstract}
}
\DCGVer{%
   \keywords{Incidences, Range searching}%
}%

\SODAVer{%
  \fancyfoot[R]{\scriptsize{Copyright \textcopyright\ 2023 by SIAM\\
  Unauthorized reproduction of this article is prohibited}}
}

\maketitle

\section{Introduction}

\DCGVer{\subparagraph{Problem statement}}%
\RegVer{\paragraph{Problem statement}}%
Let $\PS$ be a set of $n$ points in $\Re^d$, $\BS$ be a set of $m$
objects in $\Re^d$, and let $k$ be a parameter. Assume that there are
no $k$ points of $\PS$ that are all contained in $k$ objects of
$\BS$. Formally, consider the incidence graph
\begin{equation*}
    \G%
    =%
    \G(\PS, \BS)%
    =%
    \bigl( \PS \cup \BS,\, \Set{ \p \obj}{ \p \in \PS, \obj \in \BS, \text{
          and } \p \in \obj } \big),
\end{equation*}
which we assume does not contain the graph $K_{k,k}$. The question is
how many edges $\G$ has, in the worst case. As $\p \in \obj$ is an
\emphi{incidence}, the question can be formulated as bounding the
number of incidences between the points of $\PS$ and the objects of
$\BS$. We denote the number of edges of $\G(\PS, \BS)$ by
$I(\PS,\BS)$.

\DCGVer{\subparagraph{Background}}%
\RegVer{\paragraph{Background.}}%
\SODAVer{\subsubsection*{Background.}} %
The above problem is interesting because the number of incidences of
$n$ points and $n$ lines is bounded by $\Theta(n^{4/3})$, in the worst
case, and this graph avoids $K_{2,2}$, as any two lines containing the
same two distinct points are identical. The bound on the incidences of
lines and points has a long history; see \cite{s-cnhep-97} and
references therein.  Incidence problems are among the most central
classes of problems in combinatorial geometry and are closely related
to algorithmic problems in computational geometry such as range
searching~\cite{AgarwalE99} (and the so-called ``Hopcroft's
problem''~\cite{AgarwalE99,ChanZ22}).

Fox \etal \cite{fpssz-savzp-17} studied the above problem where the
objects are defined by semi-algebraic sets. For halfspaces, their
bound is $O\bigl((mn)^{d/(d+1)+\eps} + n + m \bigr)$ for any constant
$k$ (however they are studying a more general problem). They also
showed that a bound of $O(m n^{1-1/d} + n)$ holds if the set system
has \VC dimension $d$ (with $n$ elements and $m$ sets). Somewhat
surprisingly, Janzer and Pohoata \cite{jp-zpgbv-21} improved this
bound to $o(n^{2-1/d})$ for $m=n$ and $k\ge d>2$.  See also follow-up
papers by Do~\cite{Do19} and Frankl and Kupavskii~\cite{FranklK},
which gave improvements on the hidden dependence on $k$ in Fox \etal's
bounds.

There was also recent interest in the simpler version of the problem
where the ranges are axis-aligned boxes.  Basit \etal
\cite{bcstt-zpsh-21} showed an upper bound of $O(n\log^4 n)$ for $m=n$
in two dimensions, and $O(n \log^{2d} n)$ for $d>2$, assuming constant
$k$.  They proved a lower bound of $\Omega(n\log n/\log\log n)$ in two
dimensions for $m=n$ and $k=2$.  They also studied the case where the
objects are polytopes formed by the intersection of $s$ translated
halfspaces, where they show a bound of $O(n \log^s n)$.
Independently, Tomon and Zakharov \cite{tz-ttrib-21} showed a bound of
$O(k n \log^{2d+3}n)$ for the case $m=n$ (of boxes) in $\Re^d$. They
also showed a better $O(n\log n)$ bound for the special case $d=2$ if
the intersection graph avoids $K_{2,2}$.

\RegVer{\paragraph{Our results.}} %
\DCGVer{\paragraph{Our results.}} %
\SODAVer{\subsubsection*{Our results.}} %
We obtain a plethora of new results as summarized in \tblref{results}.
To simplify discussion below, we focus on the main setting when $m=n$
and $k$ is constant:

\begin{itemize}
    \smallskip%
    \item For the axis-aligned rectangles/boxes case, we prove the
    bound $O(n(\log n/\log\log n)^{d-1})$.  Not only does this improve
    the previous bounds by Basit \etal~\cite{bcstt-zpsh-21} and Tomon
    and Zakharov~\cite{tz-ttrib-21} by more than a $\log^d n$ factor,
    but our bound is also \emph{tight} for $d=2$ since it matches
    Basit \etal's $\Omega(n \log n/\log\log n)$ lower bound.  It is
    also tight for all larger $d$, as we observe in \apndref{lb} that
    an $\Omega( n(\log n/\log\log n)^{d-1} )$ lower bound for all
    $d\ge 2$ actually follows from a proof by
    Chazelle~\cite{Chazelle90a} on a seemingly different topic (data
    structure lower bounds), three decades before Basit \etal's paper.

    \medskip%
    \item For halfspaces in $d>3$ dimensions, we prove an
    $O(n^{2\floor{d/2}/(\floor{d/2}+1)})$ bound.  Not only is this a
    significant improvement over the previous $O(n^{2d/(d+1)+\eps})$
    bound, and it is \emph{tight} for $d=5$ since there is a matching
    $\Omega(n^{4/3})$ lower bound for 5D halfspaces; see \remref{5d}.
    (The dependence on $k$ in our bound is also better than in some
    previous results.)

    \medskip%
    \item For halfspaces in 2 or 3 dimensions, we prove the first
    $O(n)$ bound, which is linear (and thus obviously tight).  The 3D
    halfspace case is particularly important, since 2D disks reduce to
    3D halfspaces by a standard lifting transformation.

    \medskip%
    \item For well-behaved shapes in the plane that have linear union
    complexity (for example, pseudo-disks), we obtain a bound that is
    very close to linear (and thus very close to tight), namely,
    $O(n\log\log n)$.  For fat triangles in the plane, our bound is
    \emph{extremely} close to linear, namely, $O(n(\log^*n)^2)$ (i.e.,
    at most two iterated logarithmic factors away from tight).
\end{itemize}

\DCGVer{\subparagraph{Our techniques}} %
\RegVer{\paragraph{Our techniques.}} %
\SODAVer{\subsubsection*{Our techniques.}} %
The combinatorial problem here is closely related to the algorithmic
problem of range searching, and techniques developed for the latter
will be useful here.  The connection with range searching can be
formally explained via the notion of \emph{biclique covers} (see
\cite{Do19} or the beginning of \secref{halfspaces}), which readily
yielded an $O(n^{2d/(d+1)+\eps})$ bound for halfspaces and an
$O(n\log^d n)$ bound for axis-aligned boxes.  To get better bounds,
however, we will solve the problem directly.

For rectangles and boxes, we use simple geometric divide-and-conquer
argument, inspired by \emph{range trees} \cite{bcko-cgaa-08}, to
obtain an $O(n\log^{d-1}n)$ bound.  For halfspaces, the key intuition
is that there cannot be too many deep points when the incidence graph
avoids $K_{k,k}$; the precise quantitive claim is stated in
\lemref{level}, which we prove using \emph{shallow cuttings} (first
introduced by \Matousek \cite{m-rph-92} for halfspace range
reporting). See \defref{shallow:cutting} for details on shallow
cuttings.  Now that we know most points are shallow, we can use
shallow cuttings again to bound the number of incidences by geometric
divide-and-conquer.

For halfspaces in 3D, this approach unfortunately generates an extra
logarithmic factor.  We present a variant of the approach using a
different, interesting recursion in which we reduce the number of
points by a fraction.  By duality, we can similarly reduce the number
of halfspaces by a fraction.  In combination, the recurrence then
yields a geometric series, summing to a linear bound.

Going back to rectangles and boxes, our shallow-cutting-based approach
for 3D halfspaces can also be applied to obtain a linear bound for 2D
3-sided rectangles.  Combined with geometric divide-and-conquer with a
nonconstant branching factor, we can then improve the $O(n\log n)$
bound for 2D general axis-aligned rectangles to the tight bound of
$O(n\log n/\log\log n)$, and similarly improve the $O(n\log^{d-1}n)$
bound for boxes in $\Re^d$ to $O(n(\log n/\log\log n)^{d-1})$.

For shapes with low union complexity in 2D, we can also adapt the
shallow cutting approach, but because of the lack of duality, we get a
different bound with an extra $\log\log n$ factor. As a reminder, the
union complexity is the descriptive complexity of the union region,
see \cite{aps-sugo-08} for details.  For fat triangles in 2D, one of
the $\log^*n$ factors comes from known bounds on the union complexity
by Aronov \etal~\cite{abes-ibulf-14}.  On the other hand, the second
$\log^*n$ factor arises from an entirely different reason,
interestingly, from the way we use shallow cuttings.  Along the way,
the ideas in our proof yield a new data structure for range reporting
for 2D fat triangles, which may be of independent interest.

The interplay between combinatorial geometry and computational
geometry is nice to see, and is one of our reasons for studying this
class of problems in the first place.  In some sense, the incidence
problems here are ``cleaner'' abstractions of the algorithmic problem
of range searching, as one does not have to be concerned with
computational cost of certain operations such as point location.  But
there are subtle differences, giving rise to different bounds.  The
connection to range searching is also what led us to realize that a
lower bound argument in Chazelle's paper~\cite{Chazelle90a} had
addressed the same combinatorial problem.

\DCGVer{%
\newcommand{\mcX}[1]{%
   \setlength{\fboxrule}{0.0pt}%
   \!\begin{minipage}{5.3cm} \fbox{%
      \begin{minipage}{5.1cm}
             \begin{equation*}%
                 \hspace*{-0.1cm}%
                 \mathclap{
                    \begin{aligned}
                      #1
                    \end{aligned}
                 }
             \end{equation*}%
         \end{minipage}
      }%
  \end{minipage}
} }%
\RegVer{%
   \newcommand{\mcX}[1]{%
      \ensuremath{#1} }%
}%

\begin{figure}[t]
    \centering%
    \SODAVer{\small}%
    \begin{tabular}{|*{4}{c|}}
      \hline
      Dim
      &
        Objects
      &
        Bound
      &
        ref
      \\
      \hline
      \hline
      $d=1\Bigr.$
      &
        intervals
      &
        $\leq kn + 3km$
      &
        \lemref{intervals}%
      \\
      \hline
      $d>1$
      &
        $\begin{array}{c}
          \text{axis-aligned}\\
          \text{boxes}
         \end{array}$
      &
        \DCGVer{%
        \mcX{
        O\bigl( k n(\log n/\log\log n)^{d-1}
        \qquad\\
      +\ km\log^{d-2+\eps}n\bigr)
      }}%
      \RegVer{%
      $O\bigl( k n(\log n/\log\log n)^{d-1} +\
      km\log^{d-2+\eps}n\bigr)$
      }
      &
        \thmref{box:improved}%
      \\
      \hline
      \hline
      $d>1$
      &
        $\delta$-polyhedra
      &
        \mcX{\ensuremath{O\bigl(k n(\log n/\log\log n)^{\delta-1}}
        \DCGVer{\qquad\\}
      \ensuremath{+\  km\log^{\delta-2+\eps}n\bigr)}
      \RegVer{\ensuremath{\Bigr.}}
      }
      &
        \lemref{s:polytopes}/\thmrefs{box:improved}%
      \\
      \hline\hline
      $d=2,3$
      &
        halfplanes
      &
        $O\pth{ k(n+m) }\Bigr.$
      &
        \lemref{s:c:halfspaces:3d}%
      \\
      \hline
      $d>3$
      &
        halfspaces
      &
        \mcX{
        O\Bigl(k^{2/(\floor{d/2}+1)}
        (mn)^{\floor{d/2}/(\floor{d/2}+1)}\DCGVer{\qquad\\}
      + k (n+m) \Bigr)
      \RegVer{\biggr.}
      }
      &
        \lemref{s:c:halfspaces}%
      \\
      \hline
      \hline
      $d=2$
      &
        disks
      &
        $O\pth{ k(n+m) }\Bigr.$
      &
        \corref{balls} (I)
      \\
      \hline
      $d>3$
      &
        balls
      &
        \mcX{
        O\Bigl(
        k^{2/(\ceil{d/2}+1)}(mn)^{\ceil{d/2}/(\ceil{d/2}+1)}
        \DCGVer{\quad\\}
      + k (n+m) \Bigr)
      \RegVer{\Biggr.}
      }
      &
        \corref{balls} (II)
      \\
      \hline
      \hline
      $d=2$
      &
        $
        \begin{array}{c}
          \text{shapes union}\\
          \text{complexity }\\
          \FZero(m)
        \end{array}$
      &
        \begin{math}
            O( k n +  k \FZero(m) (\log\log m +\log k)) \Bigr.
        \end{math}
      &
        \thmref{incidences}%
      \\
      \hline
      $d=2$
      &
        pseudo-disks
      &
        \begin{math}
            O( kn + km (\log\log m + \log k)) \Bigr.
        \end{math}
      &
        \corref{p:disks}%
      \\
      \hline
      $d=2$
      &
        fat triangles
      &
        \mcX{%
        O( kn + km (\log^*m)(\log^*m \DCGVer{\qquad\\}
      + \log\log k)) \Bigr.
      }
      &
        \corref{fat:triangles}%
      \\
      \hline
    \end{tabular}
    \smallskip%
    \captionof{table}{Results so far where an induced $K_{k,k}$ is
       forbidden.  A $\delta$-polyhedra is the intersection of
       halfspaces, where each of their boundaries is orthogonal one of
       $\delta$ prespecified directions (they can be unbounded). Thus
       an axis-aligned box in $d$ dimensions is a $d$-polyhedra.  }
    \tbllab{results}
\end{figure}

\section{Intervals, Rectangles and Boxes}

\subsection{Intervals}

\begin{lemma}
    \lemlab{intervals}%
    Let $\PS$ be a set of $n$ points, and let $\BS$ be a set of $m$
    intervals on the real line. If the graph $\G(\PS,\BS)$ does not
    contain $K_{k,k}$, then the maximum number of incidences between
    $\PS$ and $\BS$ is $I_1(n,m) \leq kn + 3km$.
\end{lemma}
\begin{proof}
    Assume the points of $\PS$ are $\p_1 < \p_2 < \cdots < \p_n$. We
    break $\PS$ into $N = \ceil{n/k}$ sets, where the $i$\th set is
    \begin{equation*}
        \PS_i = \{\p^{}_{k(i-1)+1}, \ldots, \p^{}_{k(i-1)+k} \},
    \end{equation*}
    for $i=1,\ldots, N$, where all the sets contain $k$ points, except
    the last set which might contain fewer points.  Now, let
    \begin{math}
        \BS_i = \Set{  \I \in \BS }{\I \cap \PS_i \neq \emptyset },
    \end{math}
    and $m_i = \cardin{\BS_i}$. For $i < N$, at most $k-1$ intervals
    of $\BS_i$ contains all the points of $\PS_i$, and all other
    intervals of $\BS_i$ must have an endpoint in
    $\CHX{\PS_i} =[ \min \PS_i, \max \PS_i]$. In particular, let $e_i$
    be the number of intervals of $\BS_i$ with an endpoint in
    $\CHX{\PS_i}$ that do not contain this interval fully.  We have
    \begin{align*}
      I(\PS,\BS)%
      &=%
        {\sum_{i=1}^{N-1}} I( \BS_i, \PS_i)
        + I(\BS_N, \PS_N)
      \leq%
      \smash{\sum_{i=1}^{N-1}} \bigl( e_i(k-1) + (k-1)k \bigr)
      + \cardin{\PS_N} m
      \\&
      \leq
      2m( k-1) + k(k-1)(N-1) + k m
      <
      k(n +  3m) . %
    \end{align*}
\end{proof}

\subsection{Axis-aligned boxes\SODAVer{.}}%
\seclab{boxes}

For integers $\alpha \leq \beta$, let
$\IRY{\alpha}{\beta} = \{ \alpha, \alpha+1, \ldots, \beta\}$ denote
the \emphw{range} between $\alpha$ and $\beta$.  A \emphi{dyadic
   range} is a range $J= \IRY{s2^i}{(s+1)2^i - 1}$, for some
non-negative integers $s$ and $i$, where $i$ is the \emphw{rank} of
$J$. Such a range has two \emphw{children} -- specifically, the dyadic
ranges $\IRY{s2^i}{s2^i + 2^{i-1} - 1}$ and
$\IRY{s2^i + 2^{i-1}}{(s+1)2^i - 1}$.  Let $\D(n)$ be the set of all
dyadic ranges contained in $\IRX{n}=\IRY{1}{n}$.

The following is well known, and we provide a proof in \apndref{canon}
for completeness.

\SaveContent{\BodyLemmaCanon}{%
   Let $n > 1$ be an integer, and consider an arbitrary range
   $I = \IRY{\alpha}{\beta} \subseteq \IRX{n}$.  Let $\canonX{I}$
   denote a minimal (cardinality wise) union of \emph{disjoint} dyadic
   ranges that covers $I$.  The set $\canonX{I}$ is unique and has
   $\leq 2\log n$ ranges, where $\log$ is in base $2$.  }

\begin{lemma}
    \lemlab{canon}%
    \BodyLemmaCanon{}
\end{lemma}

A $d$-dimensional \emphi{box} is the Cartesian product of $d$
intervals (we consider here only axis-aligned boxes). We can now apply
an inductive argument on the dimension, to get a bound on the number
of incidences between points and boxes.

\begin{lemma}
    \lemlab{easy}%
    Let $\PS$ be a set of $n$ points in $\Re^d$, and let $\BS$ be a
    set of $m$ (axis-parallel) boxes in $\Re^d$. If the graph
    $\G(\PS,\BS)$ does not contain $K_{k,k}$, then the maximum number
    of incidences between $\PS$ and $\BS$ is
    $I_{d}(n,m) = O\bigl( k( n +m) \log^{d-1} n \bigr)$.
\end{lemma}

\begin{proof}
    Denote the points of $\PS$ in increased order of their $x$ value
    by $\p_1, \p_2, \ldots, \p_n$. For any range
    $\range = \IRY{\alpha}{\beta} \in \D(n)$, let
    $\PS_\range = \{ \p_\alpha, \ldots, \p_\beta\}$.  For a box
    $\rect \in \BS$, let $\PS_\rect$ be the set of all points of $\PS$
    with $x$-coordinate in the $x$-interval of $\rect$. Clearly, there
    is a range $\IRY{\alpha}{\beta} \subseteq \IRX{n}$ (potentially
    empty), such that $\PS_\rect = \PS_{\IRY{\alpha}{\beta}}$. The
    corresponding decomposition into dyadic sets of $\PS_\rect$ is
    \begin{equation*}
        \PPS( \rect )%
        =%
        \Set{ \PS_\range}{\range \in \canonX{\IRY{\alpha}{\beta}}}.
    \end{equation*}

    This gives rise to a natural decomposition of the original
    incidence graph into dyadic incidences graphs. Specifically, for
    every $\range \in \D(n)$, let $\BS_\range$ be the set of all boxes
    $\rect$ such that $\PS_\range \in \PPS(\rect)$. The incidence
    graph $\G_\range = \G(\PS_\range, \BS_\range)$, has the property
    that all the boxes of $\BS_\range$ have intervals on the $x$-axis
    that contains the $x$ coordinates of all the points of
    $\PS_\range$. In particular, one can interpret this as a $d-1$
    dimensional instance of boxes and points, by projecting the points
    and the boxes into a $(d-1)$-dimensional halfplane orthogonal to
    the $x$-axis. Thus, for $n_\range = \cardin{\PS_\range}$ and
    $m_\range = \cardin{\BS_\range}$, we have that
    $I_d(\PS_\range, \BS_\range) \leq I_{d-1}( n_\range, m_\range)$.

    A key property we need is that any number $i \in \IRY{1}{n}$
    participates in at most $\ceil{\log n}$ dyadic intervals of
    $\D(n)$. This implies that
    $\sum_{\range \in\D(n)} n_\range = n \ceil{\log n}$.  Similarly,
    by \lemref{canon}, we have
    $\sum_{\range \in\D(n)} m_\range \leq 2m {\log n}$.  For $d=2$,
    by \lemref{intervals}, we have
    \begin{align*}
      I_2(\PS,\BS)
      &\leq
        \sum_{\range \in\D(n)} I_2(\PS_\range , \BS_\range)
        \leq%
        \sum_{\range \in\D(n)} I_1(n_\range , m_\range)
        \leq%
        k \sum_{\range \in\D(n)} \pth{ n_\range  + 3 m_\range}
        \DCGVer{\\&}
      \leq%
      k \pth{n + 6 m} \ceil{\log n},
    \end{align*}
    For $d>2$, we use induction on the dimension, which implies
    \begin{equation*}
        I_d(\PS,\BS)%
        \leq%
        \sum_{\range \in\D(n)} I_{d-1}(n_\range , m_\range) \leq%
        k \sum_{\range \in\D(n)} O\pth{ (n_\range + m_\range)
           \log^{d-2} n} = %
        O\pth {k \pth{n + m}\log^{d-1} n}.
    \end{equation*}
\end{proof}

\subsection{Polytopes formed by the intersection of \TPDF{$s$}{s} %
   translated halfspaces\SODAVer{.}}

Let $H = \{h_1, \ldots, h_s\}$ be a set of $s$ halfspaces in
$\Re^d$. Let
$\mathcal{Z}(H) = \Set{ \cap_i (h_i +v_i) }{v_i \in \Re^d}$ be the
family of all polytopes formed by the intersection of translates of
the halfspaces of $H$. Let the \emphi{dimension} $\delta$ of $H$ be
the size of smallest set of vectors orthogonal to the boundary
hyperplanes of $H$ (thus, if two halfspaces of $H$ have parallel
boundaries, then they contribute only one to the dimension). Thus, the
family of $d$-dimensional boxes in $\Re^d$ is generated by a set $H$
of $2d$ halfspaces of dimension $\delta=d$.

\begin{lemma}
    \lemlab{s:polytopes}%
    Let $H$ be a set of $s$ halfspaces of $\Re^d$, with dimension
    $\delta$. Let $\PS$ be a set of $n$ points in $\Re^d$, and let
    $\BS \subseteq \mathcal{Z}(H)$ be a set of $m$ polyhedra. If the
    graph $G( \PS, \BS)$ does not contain $K_{k,k}$, then
    $I(\PS, \BS) = O(k(n+m) \log^{\delta-1} n)$.
\end{lemma}
\begin{proof}
    Let $v_1, \ldots, v_\delta$ be the vectors orthogonal to the
    boundary hyperplanes of the halfspaces of $H$. Let $f$ be the
    mapping of a point $\p \in \Re^d$ to
    $\permut{ v_1 \cdot \p, \ldots, v_\delta \cdot \p}$. Clearly, this
    maps a polytope of $\BS$ to an axis-parallel box in
    $\Re^\delta$. A point $\p$ is inside such a polytope $\iff$
    $f(\p)$ is inside the mapped box. The result now readily follows
    from \lemref{easy}.
\end{proof}

\section{Halfspaces}%
\seclab{halfspaces}

In this section, we consider the version of the incidence problem
where $\PS$ is a set of $n$ points in $\Re^d$, and $\BS$ is a set of
$m$ halfspaces in $\Re^d$.  We first mention two simple approaches
giving weaker bounds, before describing
our best approach using shallow cuttings.

\DCGVer{\paragraph{Approach I: Reduction to intervals.}} %
\RegVer{\paragraph{Approach I: Reduction to intervals.}} %
\SODAVer{\subsubsection*{Approach I: Reduction to intervals.}} %

One can reduce the problem to one-dimensional problem involving
intervals, using the result of Welzl \cite{w-stlcn-92}, which provides
a spanning path of the points, such that any hyperplane crosses it
only $O( n^{1-1/d})$ times. Then each halfplane becomes a set of
$O(n^{1-1/d})$ intervals on this spanning path, and one apply the
above bound for intervals. This yields a bound of
$O(k n + km n^{1-1/d})$ on the incidences -- since this bound is
inferior, we provide no further details.

\DCGVer{\paragraph{Approach II: Reduction to biclique covers.}} %
\RegVer{\paragraph{Approach II: Reduction to biclique covers.}} %
\SODAVer{\subsubsection*{Approach II: Reduction to biclique
      covers.}} %

We next mention a different approach, also observed by Do~\cite{Do19}, based on the following concept about compact representations of graphs:

\begin{definition}[Biclique covers]
    Given a graph $G=(V,E)$, a \emphi{biclique cover} is a collection
    of pairs of vertex subsets $\{(A_1,B_1),\ldots,(A_\ell,B_\ell)\}$
    such that $E = \bigcup_{i=1}^\ell (A_i\times B_i)$.  The
    \emphi{size} of the cover refers to
    $\sum_{i=1}^\ell (\cardin{A_i}+\cardin{B_i})$.
\end{definition}

\begin{remarks}
    Biclique covers~\cite{AgarwalAAS94} are closely related to
    \emphi{range searching}, where the goal is to build data
    structures for a reset $\PS$ of $n$ points, so that given a query
    object $\obj$, we can quickly report or count the points in
    $\PS\cap\obj$.  This idea was also used in the proof of
    \lemref{easy}.  Many known data structures for range searching
    produce a collection of so-called ``canonical subsets'' of $\PS$,
    so that for any query object $\obj$, the points in $\PS\cap\obj$
    can be expressed as a union of a small number of canonical
    subsets.  From such a data structure, we can form a biclique cover
    of $\G(\PS,\BS)$ by letting the $A_i$'s be the canonical subsets,
    and letting $B_i$ be the subset of all query objects $\obj\in\BS$
    that ``use'' the canonical subset $A_i$ in their answers.  Here,
    $\sum_i \cardin{A_i}$ roughly corresponds to the space or
    preprocessing time of the data structure, and
    $\sum_i \cardin{B_i}$ roughly corresponds to the total query time.

    For halfspace range searching, there are known data structures
    \cite{Matousek92,Matousek93} that can answer $m$ queries on $n$
    points in $O\bigl((n+m+(nm)^{d/(d+1)})\log^{O(1)}n\bigr)$ total
    time, and these data structures yield biclique covers of similar
    size.  For orthogonal range searching (where the ranges are
    axis-aligned boxes), known data structures (namely, range
    trees)~\cite{bcko-cgaa-08} yield biclique covers of size
    $O((n+m)\log^d n)$.
\end{remarks}

\begin{lemma}
    \lemlab{bicliques}%
    Let $\PS$ be a set of points, and let $\BS$ be a set of
    objects. If the graph $\G(\PS,\BS)$ has a biclique cover of size
    $X$ and does not contain $K_{k,k}$, then $I(\PS,\BS)=O(kX)$.
\end{lemma}
\begin{proof}
    Let $\{(A_1,B_1),\ldots,(A_\ell,B_\ell)\}$ be a biclique cover.
    For each $i$, we must have
    $\min\{\cardin{A_i},\cardin{B_i}\} < k$, because otherwise
    $\G(\PS,\BS)$ would contain $K_{k,k}$.  Thus,
    \begin{math}
        I(\PS,\BS)
        \leq%
        \sum_{i=1}^\ell \cardin{A_i} \cardin{B_i}%
        \leq%
        \sum_{i=1}^\ell k (\cardin{A_i} + \cardin{B_i}).
    \end{math}
\end{proof}

By the above lemma and known results on biclique covers, we
immediately obtain an $O(k(n+m)\log^d n)$ incidence bound for $n$
points and $m$ boxes in $\Re^d$; this is a logarithmic-factor worse
than the bound from the previous section (but is already better than
results from previous papers).  We also immediately obtain an
$O(k(n+m+(nm)^{d/(d+1)})\log^{O(1)}n)$ incidence bound (similar to Fox
\etal's bound~\cite{fpssz-savzp-17}) for $n$ points and $m$ halfspaces
in $\Re^d$; this improves the bound from the previous subsection, but
we will do even better (both as a function of $n$ and $m$, as well as
in terms of the dependence on $k$) in the next subsection for
halfspaces, by using specific known techniques developed for halfspace
range reporting -- namely, shallow cuttings.

\subsection{Better bounds using shallow cuttings}\seclab{shallow:cut}

\subsubsection{Preliminaries}

\begin{definition}[Duality]
    \deflab{duality}%
    The \emphi{dual hyperplane} of a point
    $\p = (\p_1, \ldots, \p_d) \in \Re^d$ is the hyperplane
    $\Dual{\p}$ defined by the equation
    $x_d = -p_d + \sum_{i=1}^{d-1} x_i p_i$. The \emphi{dual point} of
    a hyperplane $h$ defined by $x_d = a_d + \sum_{i=1}^{d-1} a_i x_i$
    is the point
    \begin{math}
        \Dual{h} = (a_1, a_2, \ldots, a_{d-1}, -a_d).
    \end{math}
\end{definition}

\begin{fact}
    \factlab{reverse}%
    Let $\p$ be a point and let $h$ be a hyperplane. Then $\p$ lies
    above $h$ $\iff$ the hyperplane $\Dual{\p}$ lies below the point
    $\Dual{h}$.
\end{fact}

Given a set of objects $T$ (e.g., points in $\Re^d$), let
$\Dual{T} = \Set{\Dual{x}}{x \in T}$ denote the dual set of objects.

\begin{definition}[Levels]
    \deflab{levels}%
    \deflab{top:bottom:lvls}%
    For a collection of hyperplanes $H$ in $\Re^d$, the (bottom)
    \emphi{level} of a point $\p \in \Re^d$, denoted by $\levelX{\p}$,
    is the number of hyperplanes of $H$ lying on or below $\p$.  The
    (bottom) \emphi{$k$-level} of $H$, denoted by $\BY{k}{H}$, is the
    closure of the set formed by the union of all the points on
    $\cup H$ with level $k$.

    The \emphw{$(\leq k)$-levels}, denoted by
    $\BY{\leq k}{H}= \cup_{i=0}^k \BY{i}{H}$, is the union of all
    levels up to $k$.
\end{definition}

By \factref{reverse}, if $h$ is a hyperplane which contains $k$ points
of $\PS$ lying on or above it (and at least one point lies on it),
then the dual point $\Dual{h}$ is a member of the $k$-level of
$\Dual{\PS}$.

\begin{definition}[Cuttings]
    Let $H$ be a set $H$ of $m$ hyperplanes in $\Re^d$. For a
    parameter $r \in [1,m]$, a \emphi{$1/r$-cutting} is a
    decomposition of $\Re^d$, %
    such that each (generalized) simplex intersects $\leq m/r$
    hyperplanes of $H$.
\end{definition}

A somewhat more refined
concept with better bounds is the following.

\begin{definition}[Shallow cuttings]
    \deflab{shallow:cutting}%
    Let $H$ be a set of $m$ hyperplanes in $\Re^d$. A
    \emphi{$k$-shallow $1/r$-cutting} is a collection of simplices such
    that:
    \begin{compactenumi}
        \smallskip%
        \item the union of the simplices covers the
        $(\leq k)$-levels of $H$ (see \defref{levels}), and
        \smallskip%
        \item each simplex intersects at most $m/r$ hyperplanes of
        $H$.
    \end{compactenumi}
    \smallskip%
    (Note, that the simplices are not necessarily interior disjoint.)
\end{definition}

Chazelle and Friedman \cite{cf-dvrsi-90} showed how to compute
$1/r$-cuttings with $O(r^d)$ simplices.  \Matousek \cite{m-rph-92} was
the first to show how to compute $k$-shallow $1/r$-cuttings of size
$O((rk/m+1)^d (m/k)^{\floor{d/2}})$.  In particular, there exist
$(m/r)$-shallow $1/r$-cuttings of size $O(r^{\floor{d/2}})$.

\subsubsection{Halfspaces for \TPDF{$d>3$}{d>3}\SODAVer{.}}%
\seclab{halfspaces:higher:d}

To obtain better incidence bounds for halfspaces, the following lemma
is the key and states that under the $K_{k,k}$-free assumption, most
points of $\PS$ are shallow, i.e., there are fewer points of $\PS$ at
higher levels.

A \emphw{upper halfspace} is one that lies vertically (in the $d$\th
dimension) above its boundary.
\begin{lemma}
    \lemlab{level}%
    Let $\PS$ be a set of $n$ points, and let $\BS$ be a set of $m$
    upper halfspaces, both in $\Re^d$. Let $H$ be the hyperplanes
    bounding $\BS$.  If the graph $\G(\PS,\BS)$ does not contain
    $K_{k,k}$, then for $r\le m/(2k)$, the number of points of $\PS$
    between the $m/r$-level and the $2m/r$-level of $H$ is at most
    $O(k\cdot r^{\floor{d/2}})$.
\end{lemma}
\begin{proof}
    Compute a $2m/r$-shallow $1/(2r)$-cutting $\Xi$ of $H$, consisting
    of $O\bigl(r^{\floor{d/2}}\bigr)$ simplices.  The simplices in
    $\Xi$ cover all points below the $2m/r$-level.  Consider a simplex
    $\nabla\in\Xi$ that contains at least one point between the
    $m/r$-level and the $2m/r$-level.  Since $\nabla$ contains a point
    of level at least $m/r$ and intersects at most $m/(2r)$
    hyperplanes, the number of hyperplanes completely below $\nabla$
    is at least $m/r-m/(2r)=m/(2r)\ge k$.  By the $K_{k,k}$-free
    assumption, we must have $\cardin{\nabla\cap \PS} < k$.
\end{proof}

We now prove our main result for halfspaces in dimensions $d>3$:

\begin{lemma}
    \lemlab{s:c:halfspaces}%
    Let $\PS$ be a set of $n$ points, and let $\BS$ be a set of $m$
    halfspaces, both in $\Re^d$ for a constant $d>3$. If the graph
    $\G(\PS,\BS)$ does not contain $K_{k,k}$, then the maximum number
    of incidences between $\PS$ and $\BS$ is
    \begin{equation*}
        O\bigl(
        k^{2/(\floor{d/2}+1)}(mn)^{\floor{d/2}/(\floor{d/2}+1)} + k\,
        (n+m)\bigr).%
    \end{equation*}
\end{lemma}
\begin{proof}
    Without loss of generality, assume that all halfspaces in $\BS$
    are upper halfspaces (since we can treat the lower halfspaces
    separately in a similar way and add the two incidence bounds).
    Let $I(n,m)$ denote the maximum number of incidences between $n$
    points and $m$ upper halfspaces, under the $K_{k,k}$-free
    assumption.  Let $r\le m/(2k)$ be a parameter.

    Let $H$ be the set of hyperplanes bounding $\BS$.  We break $\PS$
    into $O(\log r)$ classes. Specifically, a point $p \in \PS$ is in
    $\PS_0$ if $p$ is below the $m/r$-level of $H$. For $i>0$,
    $p \in \PS_i$ if it is between the $2^{i-1}m/r$-level and the
    $2^{i}m/r$-level of $H$.

    For $i > 0$, we have
    $\cardin{\PS_i} = O\bigl(k (r/2^{i})^{\floor{d/2}}\bigr)$ by
    \lemref{level}.  It follows that
    \begin{equation*}
        I(\PS_i, \BS)
        \le
        O( \cardin{\PS_i} \cdot 2^{i}m/r )
        =
        O\bigl(km (r/2^{i})^{\floor{d/2}-1}\bigr).
    \end{equation*}
    Summing over all $i>0$, yields
    $\sum_{i>0} I(\PS_i,\BS) = O\bigl(kmr^{\floor{d/2}-1}\bigr)$
    (assuming $d>3$).

    For $i=0$, compute an $m/r$-shallow $1/r$-cutting $\Xi_0$ of $H$,
    consisting of $O(r^{\floor{d/2}})$ simplices.  The simplices in
    $\Xi_0$ cover all points of $\PS_0$.  Consider a simplex
    $\nabla\in\Xi_0$ that intersects $\PS_0$.  Since $\nabla$ contains
    a point of level at most $m/r$ and intersects at most $m/r$
    hyperplanes, the number of hyperplanes intersecting or below
    $\nabla$ is $O(m/r)$.  By subdividing the simplices into subcells,
    we can ensure that each subcell contains
    $O\bigl(\bigl\lceil{{n}/{r^{\floor{d/2}}}}\bigr\rceil\bigr)$
    points of $\PS_0$, while keeping the number of subcells
    $O(r^{\floor{d/2}})$.  It follows that
    \begin{equation*}
        I(\PS_0, \BS)
        \leq%
        O(r^{\floor{d/2}})\cdot
        I\Bigl(\Bigl\lceil{\frac{n}{r^{\floor{d/2}}}}\Bigr\rceil,
        \frac{m}{r} \Bigr).
    \end{equation*}

    Thus, for any $r\le m/(2k)$,
    \begin{equation}
        I(n,m)
        \ \leq\
        O(r^{\floor{d/2}})\cdot
        I\Bigl(\Bigl\lceil{\frac{n}{r^{\floor{d/2}}}}\Bigr\rceil,
        \frac{m}{r} \Bigr) +
        O\bigl(k mr^{\floor{d/2}-1}\bigr).
        \eqlab{main}
    \end{equation}

    We first set $r=m/(2k)$ in \Eqref{main} and use the naive bound
    $I(n',2k) \leq n' \cdot 2k$ to obtain
    \begin{equation*}
        I(n,m) \leq O\bigl(kn + m^{\floor{d/2}}/k^{\floor{d/2}-2}\bigr).
    \end{equation*}
    In particular, $I(n,m)=O(kn)$ if
    $n=\Omega(m^{\floor{d/2}}/k^{\floor{d/2}-1})$.  By duality,
    $I(n,m)=I(m,n)=O(km)$ if
    $m=\Omega(n^{\floor{d/2}}/k^{\floor{d/2}-1})$.

    We next set $r$ in \Eqref{main} so that
    $m/r = (n/r)^{\floor{d/2}}/k^{\floor{d/2}-1}$. Namely,
    \begin{equation*}
        r =
        \Bigl(n^{\floor{d/2}}/(k^{\floor{d/2}-1}m)\Bigr)^{1/(\floor{d/2}^2-1)}.
    \end{equation*}
    Assuming $n\ll m^{\floor{d/2}}/k^{\floor{d/2}-1}$ and
    $m\gg n^{\floor{d/2}}/k^{\floor{d/2}-1}$, we indeed have
    $1\ll r\ll m/k$.  We then obtain
    \begin{math}
        I(n,m)\ =
        O\bigl(k^{2/(\floor{d/2}+1)}(mn)^{\floor{d/2}/(\floor{d/2}+1)}\bigr).
    \end{math}
\end{proof}

\begin{remark}
    \remlab{5d}%
    For example, for halfspaces in dimension $d=5$, the above bound is
    $O\bigl(k^{2/3}m^{2/3}n^{2/3} + k(n+m)\bigr)$.  This bound is
    tight, at least for $d=5$ and for constant $k\ge 2$, since
    $\Omega(m^{2/3}n^{2/3} + n+m)$ is a lower bound.  This follows
    from a known reduction of 2D point-line incidences to 5D
    point-halfspace incidences~\cite{Erickson95,Erickson96}: Take a
    set $P$ of $n$ points and a set $L$ of $m$ lines in the plane with
    $\Omega(m^{2/3}n^{2/3} + n+m)$ incidences~\cite{m-ldg-02}.  The
    incidence graph avoids $K_{2,2}$, and thus $K_{k,k}$ for any
    $k\ge 2$.  Observe that for a sufficiently small constant
    $\eps>0$, a point $\bigl.(p_x,p_y)\in P$ lies on a line $y=ax+b$
    in $L$ $\iff$ $(p_y - ap_x - b)^2\le \eps$, $\iff$
    $a^2p_x^2 + p_y^2 - 2ap_xp_y + 2abp_x - 2bp_y + b^2\le \eps$,
    $\iff$ the point $(p_x^2,p_y^2,p_xp_y,p_x,p_y)\in\Re^5$ lies in
    the halfspace
    \begin{equation*}
        \Set{\smash{(\xi_1,\xi_2,\xi_3,\xi_4,\xi_5)}}{\smash{ a^2\xi_1
              + \xi_2 -2a\xi_3 + 2ab\xi_4 - 2b\xi_5 + b^2\le
              \eps}\bigr.}.
    \end{equation*}
    Thus, we obtain $n$ points and $m$ halfspaces in $\Re^5$ with the
    same incidence graph.
\end{remark}

\subsubsection{Halfspaces for \TPDF{$d\le 3$}{d >=3 } \SODAVer{.}}%
\seclab{halfspaces:3d}

For $d\le 3$, the approach in \secref{halfspaces:higher:d} generates
an extra logarithmic factor (since $I(\PS_i,\BS)=O(km)$ for each $i$
and so $\sum_{i>0}I(\PS_i,\BS)=O(km\log m)$).  We propose a different
recursive strategy which eliminates the extra factor and achieves a
linear bound.

\begin{lemma}
    \lemlab{s:c:halfspaces:3d}%
    Let $\PS$ be a set of $n$ points, and let $\BS$ be a set of $m$
    halfspaces, both in $\Re^d$ for $d\le 3$. If the graph
    $\G(\PS,\BS)$ does not contain $K_{k,k}$, then the maximum number
    of incidences between $\PS$ and $\BS$ is $O\pth{ k (n+m) }$.
\end{lemma}
\begin{proof}
    We modify the proof of \lemref{s:c:halfspaces}, for the case $d=3$
    (as the bound for $d=2$ is readily implied by the 3d bound).
    Recall that, for $i > 0$, $p \in \PS_i$ if it is between the
    $2^{i-1}m/r$-level and the $2^{i}m/r$-level of the planes of
    $\BS$, and $\cardin{\PS_i} = O(k r/2^{i})$.  Thus,
    $\cardin{ \bigcup_{i>0} \PS_i } = O(kr)$, which can be made less
    than $m/2$ by setting $r=m/(ck)$ for a sufficiently large
    constant~$c$.  So,
    $I\bigl( \bigcup_{i>0} \PS_i, \BS\bigr) \le I(m/2,m)$.  For $i=0$,
    we use the trivial upper bound
    $I(\PS_0,\BS) \le O(\cardin{\PS_0} \cdot m/r) = O(kn)$.  It
    follows that
    \begin{equation*}
        I(n,m)\ \le\ I(m/2,m) + O(kn).
    \end{equation*}

    By duality, we also obtain $I(n,m)\le I(n,n/2) + O(km)$.  By
    alternating between these two recurrences, we have
    $I(m,m) \le I(m/2,m/2) + O(km)$, implying $I(m,m)=O(km)$ and
    $I(n,m) = O(k(n+m))$.
\end{proof}

\begin{remark}
    The above recursion, which uses duality to obtain a linear bound,
    is interesting.  A loosely similar recursion appeared in a recent
    paper~\cite{Chan21} on an algorithmic problem related to 2D
    pseudo-halfplane range reporting.
\end{remark}

\subsection{Disks and balls}

The mapping of a point $\p \in \Re^d$ to the point
$(\p, \normX{p}^2 )\in \Re^{d+1}$ can be interpreted as mapping balls
into halfspaces \cite{bcko-cgaa-08}. Applying
\lemref{s:c:halfspaces:3d} and \lemref{s:c:halfspaces} in $d+1$
dimensions gives the following.

\begin{corollary}
    \corlab{balls}%
    (I) Let $\PS$ be a set $n$ points, and $\BS$ be a set of $m$
    disks, both in the plane. If the graph $G(\PS,\BS)$ does not
    contain $K_{k,k}$, then $I(\PS,\BS) = O(k(n+m))$.

    \medskip%
    \noindent%
    (II) Let $\PS$ be a set $n$ points in $\Re^d$, and $\BS$ be a set
    of $m$ balls in $\Re^d$, for $d>2$. If the graph $G(\PS,\BS)$ does
    not contain $K_{k,k}$, then
    \begin{math}
        I(\PS,\BS) =%
        O\bigl( k^{2/(\ceil{d/2}+1)}(mn)^{\ceil{d/2}/(\ceil{d/2}+1)} +
        k (n+m)\bigr). \Bigr.
    \end{math}
\end{corollary}

\section{Back to rectangles and boxes}\seclab{boxes:improved}

In this section, we return to the case of rectangles and boxes,
and apply our method for 3D halfspaces from \secref{halfspaces:3d} to
slightly improve the bounds in \secref{boxes}.

First, we can obtain a linear bound for the case of 3D orthants, by
modifying the proof using known versions of shallow cuttings for
orthants~\cite{AfshaniT18}, or by directly mapping from 3D orthants to
3D halfspaces~\cite{ChanLP11,PachT11}: Without loss of generality,
assume that the orthants are of the form
$(-\infty,q_x]\times (-\infty,q_y]\times (-\infty,q_z]$ (since we can
treat the other 7 types of orthants separately in a similar way, and
add the bounds).  Assume that all coordinates of the points and
orthants are all integers (for example, by replacing coordinate values
with ranks).  Then the point $(p_x,p_y,p_z)$ is in the orthant
$(-\infty,q_x]\times (-\infty,q_y]\times (-\infty,q_z]$ iff the point
$(4^{p_x},4^{p_y},4^{p_z})$ is in the halfspace
$x/4^{q_x} + y/4^{q_y} + z/4^{q_z}\le 3$.

\begin{corollary}
    \corlab{orthants}%
    Let $\PS$ be a set $n$ points, and $\BS$ be a set of $m$ orthants,
    both in $\Re^3$. If the graph $G(\PS,\BS)$ does not contain
    $K_{k,k}$, then $I(\PS,\BS) = O\bigl(k(n+m)\bigr)$.
\end{corollary}

We can also obtain a linear bound for the case of 2D 3-sided
(axis-aligned) rectangles, by modifying the proof using known
(simpler) versions of shallow cuttings for 2D 3-sided
rectangles~\cite{JorgensenL11}, or by directly mapping from 2D 3-sided
rectangles to 3D orthants: Without loss of generality, assume that the
3-sided rectangles are of the form $[a,b]\times (-\infty,h]$ (since we
can treat the other 3 types of 3-sided rectangles separately in a
similar way, and add the bounds).  Then the point $(p_x,p_y)$ is in
the 3-sided rectangle $[a,b]\times (-\infty,h]$ iff the point
$(-p_x,p_x,p_y)$ is in the orthant
$(-\infty,-a]\times (-\infty,b]\times(-\infty,h]$.

\begin{corollary}
    \corlab{3sided}%
    Let $\PS$ be a set $n$ points, and $\BS$ be a set of $m$ 3-sided
    rectangles, both in $\Re^2$. If the graph $G(\PS,\BS)$ does not
    contain $K_{k,k}$, then $I(\PS,\BS) = O\bigl(k(n+m)\bigr)$.
\end{corollary}

\begin{remark}
    Alternatively, \corref{3sided} can be derived from the following
    known result by Fox and Pach~\cite{FoxP08}, and reproved by
    Mustafa and Pach~\cite{MustafaP16}: if the intersection graph of
    $N$ line segments in the plane does not contain $K_{k,k}$ for a
    constant $k$, then the graph has $O(N)$ edges.  To see the
    connection, assume that the 3-sided rectangles in $\BS$ are
    unbounded from below.  For each 3-sided rectangle in $\BS$, we
    take the horizontal line segment bounding its top side, and for
    each point in $\PS$, we take the vertical upward ray from the
    point.  The resulting intersection graph with $N=n+m$ corresponds
    to the incidence graph (since there are no crossings between two
    horizontal segments, nor between two vertical rays).  The
    dependence on $k$ was not explicitly analyzed in the previous
    papers, but also appears to be $O(k)$.
\end{remark}

Using this result for 2D 3-sided (axis-aligned) rectangles, we obtain
an improved result for 2D general 4-sided (axis-aligned) rectangles,
by using a divide-and-conquer with a nonconstant branching factor $b$.
Although the improvement is a small $\log\log n$ factor, the bound is
tight for $m=n$ and constant $k$, as it matches Basit \etal
\cite{bcstt-zpsh-21} lower bound $\Omega(n\log n/\log\log n)$.

\begin{lemma}
    \lemlab{rect:improved}%
    Let $\PS$ be a set of $n$ points in $\Re^2$, and let $\BS$ be a
    set of $m$ rectangles in $\Re^2$. If the graph $\G(\PS,\BS)$ does
    not contain $K_{k,k}$, then the maximum number of incidences
    between $\PS$ and $\BS$ is
    $I(n,m)=O\left( k n\log n/\log\log n + km\log^\eps n\right)$ for
    any constant $\eps>0$.
\end{lemma}
\begin{proof}
    Divide the plane into $b$ vertical slabs
    $\sigma_1,\ldots,\sigma_b$, each containing $n/b$ points of $\PS$.
    For each $i=1,\ldots,b$, let $\PS_i=P\cap\sigma_i$, let $\BS_i$ be
    the set of all rectangles of $\BS$ that are completely inside
    $\sigma_i$, and let $\BS'_i$ be the set of all rectangles of $\BS$
    that intersect $\sigma_i$ but are not completely inside
    $\sigma_i$.  Inside $\sigma_i$, the rectangles in $\BS'_i$ may be
    viewed as 3-sided, and so
    \begin{math}
        I(\PS_i,\BS'_i)%
        =%
        O \bigl( k(\cardin{\PS_i} + \cardin{\BS'_i} ) \bigr)
    \end{math}
    by \corref{3sided}.  Thus,
    $\sum_{i=1}^b I(\PS_i,\BS'_i)=O\bigl(k(n+bm_0)\bigr)$ where
    $m_0 = \cardin{\bigcup_{i=1}^b \BS'_i }$.  We also have
    $\sum_{i=1}^b I(\PS_i,\BS_i) \leq \sum_{i=1}^b I(n/b,m_i)$ where
    $m_i=\cardin{\BS_i}$.  We obtain the recurrence
    \begin{equation*}
        I(n,m)\ \leq%
        \max_{m_0,\ldots,m_b:\ m_0+\cdots+m_b\le m}
        \Bigl({\sum_{i=1}^b} I(n/b,m_i) + O(kn+bkm_0)\Bigr),
    \end{equation*}
    which solves to $I(n,m)\le O(kn\log_b n + bkm)$.  In particular,
    setting $b=\log^\eps n$ gives the result.
\end{proof}

The improvement in 2D also implies an improvement for boxes in higher dimensions.  Again, we can use a divide-and-conquer with a larger
branching factor $b$.

\begin{theorem}
    \thmlab{box:improved}%
    Let $\PS$ be a set of $n$ points in $\Re^d$, and let $\BS$ be a
    set of $m$ boxes in $\Re^d$. If the graph $\G(\PS,\BS)$ does not
    contain $K_{k,k}$, then the maximum number of incidences between
    $\PS$ and $\BS$ is
    $O\bigl( k n (\log n/\log\log n)^{d-1} + km\log^{d-2+\eps} n\bigr)$.
\end{theorem}
\begin{proof}
    Assume that $\PS$ is a set of $n$ points inside a vertical slab
    $\sigma$, $\BS$ is a set of $m$ boxes each having at least one
    vertex in $\sigma$.  Let $v$ be the number of vertices of the
    boxes that are inside $\sigma$.  Then $v\le 2^d m$.  Let
    $I_d(n,v)$ denote the maximum number of point-box incidences in
    this setting.

    Divide $\sigma$ into $b$ vertical slabs
    $\sigma_1,\ldots,\sigma_b$, each containing $n/b$ points of $\PS$.
    For each $i=1,\ldots,b$, let $\PS_i=P\cap\sigma_i$, let $\BS_i$ be
    the set of all boxes of $\BS$ that have at least one vertex in
    $\sigma_i$, and let $\BS'_i$ be the set of all ``long'' boxes of
    $\BS$ that intersect $\sigma_i$ but are not in $\BS_i$ (i.e., that
    completely cut across $\sigma_i$).  Inside $\sigma_i$, the boxes
    in $\BS'_i$ are liftings of $(d-1)$-dimensional boxes, and so
    $I(\PS_i,\BS'_i)\le I_{d-1}(n/b,v)$.  We can recursively bound
    $\sum_{i=1}^b I(\PS_i,\BS_i)$ by $\sum_{i=1}^b I_d(n/b,v_i)$ where
    $v_i$ is the number of vertices in $\sigma_i$.  We obtain the
    recurrence
    \[
        I_d(n,v) \ \leq%
        \max_{v_1,\ldots,v_b:\ v_1+\cdots+v_b\le v}
        \Bigl(\smash{\sum_{i=1}^b} I(n/b,v_i) + bI_{d-1}(n/b,v)\Bigr),
    \]
    with the base case $I_2(n,v)=O(kn\log_b n + bkv)$ from
    \lemref{rect:improved}.  This solves to
    \begin{equation*}
        I_d(n,m)\le
        O\bigl( kn(\log_b n)^{d-1} + b^{d-1}kv(\log_b n)^{d-2} \bigr).
    \end{equation*}
    Setting $b=\log^{\eps/(d-1)} n$ gives the result.
\end{proof}

We immediately obtain a similar improvement of \lemref{s:polytopes} for
$\delta$-polyhedra.

\begin{remark}
    Similar ideas of using shallow cuttings for 3D orthants and 2D
    3-sided rectangles, reducing the 4-sided to the 3-sided case, and
    doing divide-and-conquer with $\log^\eps n$ branching factors can
    also be found in the algorithmic literature on orthogonal range
    searching.  The $\log n/\log\log n$ factors in
    \lemref{rect:improved} may look surprising at first, but it is not
    unprecedented in the range searching literature. For example,
    Chazelle~\cite{Chazelle90a} proved an $\Omega(n (\log n/\log\log n)^{d-1})$
    space lower bound for orthogonal range reporting in pointer machines,
    and in \apndref{lb}, we note that his argument actually
    implies a lower bound for our incidence problem.
\end{remark}

\begin{remark}
    It is possible to improve the bound to
    $O(n(\log n/\log\log n)^{d-1} + m\log^{d-3+\eps}n)$ for boxes in
    dimension $d\ge 3$ and for constant $k$.  The idea is to prove an
    $O(n\log_b n + bm\log\log n)$ bound for the 3D 5-sided case by
    employing a known recursive grid approach introduced by Alstrup
    \etal~\cite{AlstrupBR00} and adapted by Chan
    \etal~\cite{ChanNRT18}, using the bound for 3D orthants as the
    base case.  However, this does not imply any improvement for the
    main case $m=n$, and so we will omit the details.
\end{remark}

\section{When the union %
   complexity is low}
\seclab{union:complexity}

In this section, we consider the version of the incidence problem for
set $\PS$ of $n$ points and a set $\Family$ of $m$ well-behaved shapes
with near-linear union complexity.  We will adapt the shallow cutting
approach from \secref{shallow:cut}.

\subsection{Preliminaries} %

\begin{definition}
    \deflab{well:behaved}%
    Let $\Family$ be a set of $\nF$ shapes in $\Re^d$. The family
    $\Family$ is \emphi{well-behaved} if
    \begin{compactenumI}
        \smallskip%
        \item The union complexity for any subset of $\Family$ of size
        $r$ is (at most) $\FZero(r)$, for some function
        $\FZero(r)$ such that $\FZero(r)/r$ is nondecreasing.

        \smallskip%
        \item The total complexity of an arrangement of any $r$ shapes
        of $\Family$ is $O(r^d)$.

        \smallskip%
        \item For any set $X \subseteq \Re^d$, and any subset
        $\FamilyB \subseteq \Family$, one can perform a decomposition
        of the cells of the arrangement $\ArrX{\FamilyB}$ that
        intersects $X$, into cells of constant descriptive complexity
        (e.g., vertical trapezoids), and the complexity of this
        decomposition is proportional to the number of vertices of the
        cells of $\ArrX{\FamilyB}$ that intersects $X$.

        \smallskip%
        \item The arrangement of any $d$ shapes of $\Family$ has
        constant complexity.
    \end{compactenumI}
\end{definition}

The following minor variant of \Matousek's shallow cuttings
\cite{m-rph-92} is from Chekuri \etal \cite{cch-smcpg-12}.

\begin{theorem}
    \thmlab{shallow:cutting}%
    Given a set $\Family$ of $\nF$
    \defrefY{well-behaved}{well:behaved} shapes in $\Re^d$, and
    parameters $r$ and $k$, one can compute a decomposition $\Cutting$
    of space into $O(r^d)$ cells of constant descriptive complexity,
    such that total weight of boundaries of shapes of $\Family$
    intersecting a single cell is at most $\nF/r$. The decomposition
    $\Cutting$ is a \emphi{$1/r$-cutting} of $\Family$.  Furthermore,
    the total number of cells of $\Cutting$ containing points of depth
    smaller than $k$ is
    \begin{math}
        O\bigl( ({r k}/{\nF} + 1)^d\, \FZero({\nF}/{k}) \bigr)\Big..
    \end{math}
    Here, the \emphi{depth} of a point $p$ is the number of shapes in
    $\Family$ that contains $\p$.
\end{theorem}

\begin{definition}
    A set $\Family$ of $\nF$ shapes has \emphi{primal shatter
       dimension} $b$ if for every point set $\PS$, we have
    $\cardin{\Set{\PS\cap \obj}{ \obj\in\Family}} =
    O(\cardin{\PS}^b)$.
\end{definition}

\begin{lemma}
    \lemlab{shatter}%
    Let $\PS$ be a set of $n$ points, and let $\Family$ be
    a set of $m$ shapes with primal shatter dimension~$b$.  If the graph
    $\G(\P,\Family)$ does not contain $K_{k,k}$, then the maximum
    number of incidences between $\P$ and $\Family$ is
    $I(n,m) = O( k(n^{b+1} + m))$.
\end{lemma}
\begin{proof}
    Let $\Family_0$ be the set of all shapes of $\Family$ that contain
    more than $k$ points.  Call two shapes $\obj$ and $\obj'$
    equivalent if $\PS\cap\obj = \PS\cap\obj'$.  There are $O(n^b)$
    equivalence classes of shapes.  Each equivalence class contains at
    most $k$ shapes in $\Family_0$, by the $K_{k,k}$-free
    assumption. Thus, $\cardin{\Family_0} \leq O(k n^b)$, and
    $I(\P,\Family_0)\le O(k n^{b+1})$.  On the other hand,
    $I(\P,\Family-\Family_0)\le O(km)$.
\end{proof}

\begin{remark}
    If $\G(\P,\Family)$ does not contain $K_{k,k}$, it is not
    difficult to see that the primal shatter dimension must be bounded
    by $O(k)$.  However, most families of shapes encountered in
    low-dimensional geometry have $O(1)$ VC dimension, and thus by
    Sauer's lemma~\cite{m-ldg-02}, actually have $O(1)$ primal shatter
    dimension independent of $k$.
\end{remark}

\subsection{Shapes with low union complexity}

We now modify our approach for 3D halfspaces to prove incidence bounds
for 2D well-behaved shapes with low (near-linear) union complexity.
We do not have duality now, and as a result, a straightforward
adaptation of the proof of \lemref{s:c:halfspaces} would result in an
extra logarithmic factor.  We describe a nontrivial modification to
lower the extra factor to $\log\log m + \log k$.

\begin{lemma}
    \lemlab{depth}%
    Let $\PS$ be a set of $n$ points, and let $\Family$ be a set $\nF$
    well-behaved shapes in the plane, with union complexity
    $O\bigl( \FZero(\nF) \bigr)$.  If the graph $\G(\P,\Family)$ does
    not contain $K_{k,k}$, then for any $r\le m/(2k)$, the number of
    points of $\PS$ having depth between $m/r$ and $2m/r$ is at most
    $O(k\cdot \FZero(r))$.
\end{lemma}
\begin{proof}
    Compute a $1/(2r)$-cutting $\Xi$ by \thmref{shallow:cutting}.  Let
    $\Xi$ be the cells of the cutting that contain at least one point
    with depth between $m/r$ and $2m/r$; there are $O\bigl(\FZero(r))$
    such cells.  Consider a cell $\nabla\in\Xi$.  Since $\nabla$
    contains a point of depth at least $m/r$ and intersects the
    boundaries of at most $m/(2r)$ shapes, the number of shapes
    completely containing $\nabla$ is at least
    $m/r-m/(2r)=m/(2r)\ge k$.  By the $K_{k,k}$-free assumption, we
    must have $\cardin{\nabla\cap \PS} < k$.
\end{proof}

\begin{theorem}
    \thmlab{incidences}%
    Let $\PS$ be a set of $n$ points in the plane, and let $\Family$
    be a set $\nF$ well-behaved shapes in the plane, with union
    complexity $O\bigl( \FZero(\nF) \bigr)$ and constant primal
    shatter dimension.  If the graph $\G(\P,\Family)$ does not contain
    $K_{k,k}$, then the maximum number of incidences between $\P$ and
    $\Family$ is
    $I(n,\nF) = O\bigl( k n + k \FZero(\nF) (\log\log m + \log
    k)\bigr)$.
\end{theorem}
\begin{proof}
    Let $\phi(m)=\FZero(m)/m$.  Let $t_0,t_1,\ldots, t_{\ell}$ be an
    increasing sequence of parameters, with $t_0=2k$ and
    $t_\ell\ge m$.  We break $\PS$ into $O(\ell)$
    classes. Specifically, a point $p \in \PS$ is in $\PS_0$ if its
    depth is below $t_0$. For $i>0$, $p \in \PS_i$ if its depth is
    between $t_{i-1}$ and $t_i$.

    For $i > 0$, we have
    $\cardin{\PS_i} = O\bigl(k\cdot (\FZero(m/(2t_{i-1})) +
    \FZero(m/(4t_{i-1})) + \cdots)\bigr)$ by \lemref{level}; this
    gives $\cardin{\PS_i} = O(k(m/t_{i-1})\phi(m))$.  Compute a
    $t_i/m$-cutting $\Xi$ by \thmref{shallow:cutting}.  Let $\Xi$ be
    the cells of the cutting that intersect $\PS_i$; there are
    $O\bigl(\FZero(m/t_i))\le O((m/t_i)\phi(m))$ such cells.  Consider
    a cell $\nabla\in\Xi$.  Since $\nabla$ contains a point of depth
    at most $t_i$ and intersects the boundaries of at most $t_i$
    shapes, the number of shapes intersecting or containing $\nabla$
    is $O(t_i)$.  By subdividing the simplices into subcells, we can
    ensure that each subcell contains at most
    $O\left(\frac{k(m/t_{i-1})\phi(m)}{(m/t_i)\phi(m)}\right) =
    O(kt_i/t_{i-1})$ points of $\PS_i$, while keeping the number of
    subcells $O\bigl((m/t_i)\phi(m)\bigr)$.  It follows that for
    $i>0$,
    \begin{equation*}
        I(\PS_i, \BS)
        \leq
        O\bigl((m/t_i)\phi(m)\bigr)\cdot I(kt_i/t_{i-1}, t_i).
    \end{equation*}

    For $i=0$, we use the trivial upper bound
    $I(\PS_0,\BS) \le O(\cardin{\PS_0} \cdot t_0) = O(kn)$.

    It follows that
   \begin{equation}\eqlab{low:union}
       I(n,m)
       \ \leq\ %
       \sum_{i=1}^\ell
       O\bigl((m/t_i)\phi(m)\bigr)\cdot I(kt_i/t_{i-1}, t_i) + O(kn).
    \end{equation}

    We choose the sequence $t_0=2k$, $t_i=2t_{i-1}$ for
    $0<i\le c\log k$, and $t_i=t_{i-1}^{c/(c-1)}$ for $i>c\log k$ for
    some constant $c$.  The number of terms is
    $\ell=O(\log k + \log\log m)$.  For $i\le c\log k$, we have
    $I(kt_i/t_{i-1}, t_i) = I(2k,t_i) = O(kt_i)$.  For $i>c\log k$, we
    have $t_i\ge k^{c}$ and
    $I(kt_i/t_{i-1}, t_i) = I(k t_i^{1/c}, t_i) \le I(t_i^{2/c}, t_i)
    = O(kt_i)$ by \lemref{shatter} if $c$ is sufficiently large.  Then
    \Eqref{low:union} implies
    $I(n,m)=O(\ell km\phi(m) + kn)=O(km\phi(m)(\log k+\log\log m) +
    kn)$.
\end{proof}

A set $\Family$ of simply-connected shapes in the plane is a
\emphw{set of pseudo-disks}, if for every pair of shapes, their
boundaries intersect at most twice. It known that the union complexity
of pseudo-disks is linear \cite{sa-dsstg-95} and the primal shatter
dimension is $O(1)$~\cite{m-ldg-02}. For technical reasons it is convenient
to assume that the pseudo-disks are $y$-monotone -- that is, any
vertical line either avoids a pseudo-disk, or intersects it in an
interval.  We thus get the following.

\begin{corollary}
    \corlab{p:disks}%
    Let $\P$ be a set of $n$ points in the plane, and let $\Family$ be
    a set of $\nF$ $y$-monotone pseudo-disks in the plane.  If the
    graph $\G(\P,\Family)$ does not contain $K_{k,k}$, then the
    maximum number of incidences between $\P$ and $\Family$ is bounded
    by $I(n,m) = O\bigl( k n + k \nF (\log\log \nF + \log k)\bigr)$.
\end{corollary}

A natural open question is to improve the bound to linear for
pseudodisks.

\begin{remark}
    The idea of taking a sequence of $O(\log\log m)$ shallow cuttings
    was also used in some known data structures for 3D halfspace range
    reporting \cite{Chan00,r-orrrs-99} -- the above proof is partly
    inspired by these algorithmic works.
\end{remark}

\subsection{Fat triangles}

A set of triangles $\Family$ is \emphi{fat} if the smallest angle in
any triangle of $\Family$ is bounded away from $0$ by some absolute
positive constant.  Aronov \etal \cite{abes-ibulf-14} showed that
the union complexity here is $\FZero(\nF)=O(\nF\log^*\nF)$.
So, \thmref{incidences} immediately implies $I(n,m) = O( k n + k \nF \log^* \nF (\log\log\nF + \log k))$.
In this subsection,
we show how to further lower the $\log\log\nF$ factor to $\log^*\nF$ for fat triangles.

We first prove a weaker bound $I(n,m) = O(k n\log^2 n + km)$ for fat triangles.
To this end, we start with the following lemma:

\begin{lemma}
    \corlab{curtains}%
    (I) Let $\PS$ be a set $n$ points in $\Re^3$, and $\BS$ be a set
    of $m$ wedges of the form $\{(x,y,z)\mid y\le ax+b,\ z\le c\}$. If
    the graph $G(\PS,\BS)$ does not contain $K_{k,k}$, then
    $I(\PS,\BS) = O\bigl(k(n+m)\bigr)$.

    \medskip\noindent%
    (II) Let $\PS$ be a set $n$ points in $\Re^2$, and $\BS$ be a set
    of $m$ wedges of the form $\{(x,y)\mid y\le ax+b,\ x\le c\}$. If
    the graph $G(\PS,\BS)$ does not contain $K_{k,k}$, then
    $I(\PS,\BS) = O\bigl(k(n+m)\bigr)$.

    \medskip\noindent%
    (III) Let $\PS$ be a set $n$ points in $\Re^2$, and $\BS$ be a set
    of $m$ ``curtains'' of the form
    $\{(x,y)\mid y\le ax+b,\ c\le x\le c'\}$. If the graph
    $G(\PS,\BS)$ does not contain $K_{k,k}$, then
    $I(\PS,\BS) = O\bigl(k(n\log n +m)\bigr)$.
\end{lemma}
\begin{proof}
    For (I), the union complexity for $m$ wedges of this type is
    $O(m)$.  To see this, imagine moving a horizontal sweep plane from
    top to bottom.  On the sweep plane, the complement of the union is
    an intersection ${\cal I}$ of halfplanes.  As the sweep plane
    moves, we encounter $m$ insertions of halfplanes; each insertion
    creates at most two vertices to ${\cal I}$.  Thus, the total
    number of vertices created is $O(m)$.  Having established linear
    union complexity, we could now apply \thmref{incidences}, but we
    can do better by observing that we have duality here: a point
    $(p_x,p_y,p_z)$ is in the wedge
    $\{(x,y,z)\mid y\le ax+b,\ z\le c\}$ iff the point $(a,-b,-c)$ is
    in the wedge
    $\{(\alpha,\beta,\gamma)\mid \beta\le p_x\alpha-p_y,\ \gamma\le
    -p_z\}$.  So, we can adapt the proof in \lemref{s:c:halfspaces:3d}
    to obtain the $O(k(n+m))$ incidence bound.

    \smallskip%
    (II) is a special case of (I): a point $(p_x,p_y)$ is in the wedge
    $\{(x,y)\mid y\le ax+b,\ x\le c\}$ iff the point $(p_x,p_y,p_x)$
    is in the wedge $\{(x,y,z)\mid y\le ax+b,\ z\le c\}$.

    \smallskip%
    For (III), we use divide-and-conquer to reduce to (II).  Divide
    the plane into two vertical slabs $\sigma_1$ and $\sigma_2$, each
    containing $n/2$ points of $\PS$.  For each $i\in\{1,2\}$, let
    $\PS_i=P\cap\sigma_i$, and let $\BS_i$ be the set of all curtains
    of $\BS$ that are completely inside $\sigma_i$. Let $\BS_0$ be the
    set of all curtains of $\BS$ that intersect both $\sigma_1$ and
    $\sigma_2$.  Inside $\sigma_i$, the curtains in $\BS_0$ may be
    viewed as wedges, and so
    \begin{math}
        I(\PS_i,\BS_0)=O\bigl(k(\cardin{\PS_i} + \cardin{\BS_0})
        \bigr)
    \end{math}
    by (II).  Thus,
    $I(\PS,\BS)\le I(\PS_1,\BS_1) + I(\PS_2,\BS_2) +
    O\bigl(k(n+\cardin{\BS_0})\bigr)$.  We obtain the recurrence
    \[
        I(n,m)\ \le \max_{m_0,m_1,m_2:\ m_0+m_1+m_2= m}
        \left(I(n/2,m_1)+I(n/2,m_2)+ O(kn+km_0)\right),
    \]
    which solves to $I(n,m)\le O(kn\log n + km)$.
\end{proof}

The special case of triangles containing the origin can be reduced to
the case of curtains:

\begin{corollary}
    \corlab{origin:triangles}%
    Let $\P$ be a set of $n$ points in the plane, and let $\Family$ be
    a set of $\nF$ triangles in the plane, such that all triangles
    contain the origin.  If the graph $\G(\P,\Family)$ does not
    contain $K_{k,k}$, then the maximum number of incidences between
    $\P$ and $\Family$ is bounded by $O\bigl( k (n\log n+m)\bigr)$.
\end{corollary}
\begin{proof}
    Without loss of generality, assume that the triangles all have the
    origin as one of its vertices, since a triangle containing the
    origin can be decomposed into three such triangles.  Also assume
    that the points in $\P$ are all above the $x$-axis (since we can
    treat the points below the $x$-axis in a similar way, and add the
    two bounds).  The affine transformation $(x,y)\mapsto (-1/x, y/x)$
    then maps triangles with the origin as a vertex to curtains.
\end{proof}

\begin{lemma}
    \lemlab{fat:triangles}%
    Let $\P$ be a set of $n$ points in the plane, and let $\Family$ be
    a set of $\nF$ fat triangles in the plane. If the graph
    $\G(\P,\Family)$ does not contain $K_{k,k}$, then the maximum
    number of incidences between $\P$ and $\Family$ is bounded by
    $I(n,m) = O\bigl( k (n\log^2 n + m)\bigr)$.
\end{lemma}
\begin{proof}
    By rescaling, assume that all points and triangles are contained
    in $[0,1]^2$.

    A \emphi{quadtree square} is a square of the form
    $[i/2^\ell, (i+1)/2^\ell)\times [j/2^\ell, (j+1)/2^\ell)$ for
    integers $i,j,\ell$.  A shape $\obj$ with diameter $r$ is
    \emphi{aligned} if it is inside a quadtree square of side length
    $4r$.

    For any shape $\obj$ contained in $[0,1]^2$, it is
    known~\cite{Chan03} that $\obj$ is aligned, or $\obj + (1/3,1/3)$
    is aligned, or $\obj + (2/3,2/3)$ is aligned.  Without loss of
    generality, assume that all triangles in $\Family$ are aligned,
    since we can separately treat the triangles $\obj$ where
    $\obj + (1/3,1/3)$ is aligned and the triangles $\obj$ where
    $\obj + (2/3,2/3)$ (by shifting all the triangles and points) and
    add the three bounds.

    We use a balanced divide-and-conquer.  First it is
    known~\cite{AryaMNSW98} that there exists a quadtree square $s$
    that has at most $4n/5$ points inside and at most $4n/5$ points
    outside.  (Such a square $s$ can be obtained by taking a
    ``centroid'' in the quadtree; or, to be more self-contained, we
    can let $s$ be the smallest quadtree square that has at least
    $n/5$ points inside.)  Let $\PS_1$ be the set of all points of
    $\PS$ inside $s$, and $\PS_2$ be the set of all points of $\PS$
    outside $s$.  Let $\Family_1$ be the set of all triangles of
    $\Family$ completely inside $s$, $\Family_2$ be the set of all
    triangles of $\Family$ completely outside $s$, and $\Family_0$ be
    the set of all triangles of $\Family$ that intersect the boundary
    of $s$.

    Let $r$ be the side length of $s$.  Because all triangles are
    aligned, the triangles of $\Family_0$ must have diameter at least
    $r/4$.  Because the triangles are fat and intersect a square $s$
    of side length $r$, they can be stabbed by $O(1)$ points.  We can
    bound incidences for the triangles containing each of the $O(1)$
    stabbing points by invoking \corref{origin:triangles} (after
    translating the origin).  It follows that
    $I(\P,\Family_0)=O\bigl(k(n\log n +\cardin{\Family_0})\bigr)$.
    Thus,
    $I(\PS,\Family)\le I(\PS_1,\Family_1) + I(\PS_2,\Family_2) +
    O\bigl(k(n+\cardin{\Family_0})\bigr)$.
    Let
    \begin{equation*}
        M =
        \Set{(n_1,n_2,m_0,m_1,m_2)}{
           \begin{array}{l}
             n_1, n_2, m_0,m_1,m_2 \geq 0,
             \,\,
             n_1,n_2\leq 4n/5,\\
             n_1+n_2=n,\,\,
             m_0+m_1+m_2= m
           \end{array}
        }
    \end{equation*}
    We obtain the recurrence
    \begin{equation*}
        I(n,m)\ \leq \max_{(n_1,n_2,m_0,m_1,m_2) \in M}
        \left(I(n_1,m_1)+I(n_2,m_2)+ O(kn\log n +km_0)\right),
    \end{equation*}
    which solves to $I(n,m)\leq O(kn\log^2 n + km)$.
\end{proof}

\begin{corollary}
    \corlab{fat:triangles}%
    Let $\P$ be a set of $n$ points in the plane, and let $\Family$ be
    a set of $\nF$ fat triangles in the plane, such that
    $\G(\P,\Family)$ does not contain $K_{k,k}$. Then, the maximum
    number of incidences between $\P$ and $\Family$ is bounded by
    \begin{equation*}
        I(n,m) = O\bigl( k n + k \nF (\log^*\nF+\log\log k)
        \log^* \nF \bigr).
    \end{equation*}
\end{corollary}
\begin{proof}
    We follow the proof of \thmref{incidences}, but use a sparser
    sequence: $t_0=2k$, $t_i=2t_{i-1}$ for $0<i\le i_0$, and
    $t_i=2^{\sqrt{t_{i-1}/k}}$ for $i>i_0$, with
    $i_0=\ceil{3\log\log k}$.  The number of terms is
    $\ell=O(\log\log k + \log^* m)$.  (To see this, observe that
    $t_{i_0}\ge k\log^3k$, and $t_{i_0+1}\ge k^{\omega(1)}$, and
    $t_i \ge 2^{t_{i-1}^{1/2-o(1)}}$ for $i>i_0+1$.)  For $i\le i_0$,
    we have $I(kt_i/t_{i-1}, t_i) = I(2k,t_i) = O(kt_i)$.  For
    $i>i_0$, by \lemref{fat:triangles} we have
    $I(kt_i/t_{i-1}, t_i) = O(k\cdot ((kt_i/t_{i-1})\log^2 t_i + t_i))
    = O(kt_i)$.  Then \Eqref{low:union} implies
    $I(n,m)=O(\ell km\phi(m) + kn)=O(km\phi(m)(\log\log k+\log^* m) +
    kn)$, where $\phi(m)=O(\log^*m)$ by Aronov
    \etal~\cite{abes-ibulf-14}.
\end{proof}

\begin{remark}
    As noted, some of our proofs for the incidence problem are related
    to known data structure techniques for range searching.
    Interestingly, the ideas in the above proof actually imply a new
    data structure for the problem of range reporting for points in 2D
    in the case where the query ranges are fat triangles:

    First, Chazelle and Guibas~\cite{ChazelleG86ii} have already given
    a data structure for range reporting for the case of curtains in
    2D (they called this case ``slanted range search'') with $O(n)$
    space and $O(\log n + K)$ query time, where $K$ denotes the output
    size.  The same bound thus holds for 2D triangle range reporting
    where the query ranges are triangles containing the origin
    (incidentally, this improves a result from~\cite{IshaqueSB08}).
    The quadtree-based approach in the proof of \lemref{fat:triangles}
    then gives a data structure for 2D fat-triangle range reporting
    with $O(n\log n)$ space and $O(\log n+K)$ query time.  This result
    may be of independent interest (see \cite{SharirS11} for previous
    results which have a larger query time, though with linear space).
\end{remark}

\begin{lemma}
    Let $\P$ be a set of $n$ points in the plane.  We can construct a
    data structure with $O(n\log n)$ space, so that we can report all
    $K$ points inside any query fat triangle in $O(\log n + K)$ time.
\end{lemma}

\section{Final Remarks}

It is not clear how to get a near linear bound on the number of
incidences between points and unit balls in $\Re^3$ when avoiding
$K_{k,k}$ (\corref{balls} implies a bound of $O(k^{2/3}n^{4/3})$).  In
particular, consider the following variant of the unit distance
problem, which we leave as an open problem for further research.

\begin{problem}[At most unit distance when avoiding $K_{k,k}$]
    Let $\PS$ and $\PS'$ be two sets of $n$ red and blue points,
    respectively, in $\Re^d$ for $d>2$. What is the maximum number of
    bichromatic pairs of points with distance $\leq 1$, under the
    assumption that the bipartite graph $G(\PS, \PS')$, connecting all
    pairs of points in distance $\leq 1$, does not contain $K_{k,k}$?
\end{problem}

(The non-bichromatic version of this problem has a linear bound, as
the point set cannot have balls of radius $1/2$ that contain $k$
points.)  The $\Omega(n^{4/3})$ lower bound for 5D halfspaces from
\remref{5d} also applies to 5D unit balls (since we can replace
halfspaces with balls of a sufficiently large radius, and then rescale
the balls and the points), but our upper bound for 5D unit balls from
\corref{balls} is only $O(k^{1/2}n^{3/2})$.

Another open question is whether lower bounds could be proved to show
tightness of our result for halfspaces in dimensions other than 2, 3,
and 5.

\DCGVer{\bmhead{Acknowledgments}}
\RegVer{\paragraph*{Acknowledgments.}}
\SODAVer{\subsubsection*{Acknowledgments.}}

The authors thank Abdul Basit and Adam Sheffer for useful discussions
on this problem, and comments on this manuscript.

Work by the first author on this paper was partially supported by an
NSF AF award CCF-2224271.  Work by the second author on this paper was
partially supported by an NSF AF award CCF-1907400.

\BibTexMode{%
   \RegVer{\bibliographystyle{alpha}}%
   \bibliography{klan}%
}%
\BibLatexMode{\printbibliography}

\appendix

\section{Proof of \TPDF{\lemref{canon}}{Lemma 2}}
\apndlab{canon}

\RestatementOf{\lemref{canon}}%
{ \BodyLemmaCanon{} }

\begin{proof}
    This is well known, and we include a proof for the sake of
    completeness.  Assume that $n = 2^h$. A dyadic range of $\D(n)$ is
    a \emphw{boundary} range if it is not contained in $I$, but one of
    its children is contained in $I$.  Observe that there can be at
    most two boundary ranges of the same rank, as they each cover an
    endpoint of $I$.  Replacing each boundary range by its child
    included in $I$, results in the desired decomposition $\canonX{I}$
    of $I$ into dyadic ranges (this fails only for the case where
    $I = \IRX{n}$ -- but then $\canonX{\IRX{n}}$ is a single range).

    We are left with bounding the number of boundary ranges when
    $I \subsetneq \IRX{n}$.  Observe that a range of rank $0$ can not
    be a boundary range. Furthermore, there is only a single range of
    rank $h$, specially $\IRX{n}$, and if it is not a boundary range,
    then the number of boundary ranges is at most $ 2(h-1) = 2h
    -2$. If $\IRX{n}$ is a boundary range, then there is at most one
    range of each rank that is a boundary range (since the other
    endpoint of $I$ is ``unreachable'' to be covered by an other
    boundary range), which implies a bound of $h$ on the number of
    boundary ranges in this case. Thus the number of boundary ranges
    is bounded by $\max(2h -2, h) \leq 2h-2$.

    As for the uniqueness of $\canonX{I}$, observe that any dyadic
    range properly containing a range of $\canonX{I}$, must also
    contain one of the endpoints of $I$ in its ``interior'', which
    readily implies the minimality of $\canonX{I}$, as any other
    dyadic decomposition of $I$ must refine the ranges in
    $\canonX{I}$.

    Finally, for $n$ that is not a power of two, observe that
    $2 \ceil{\log n} -2 \leq 2 \log n$.
\end{proof}

\section{On Chazelle's Lower Bound}
\apndlab{lb}

Chazelle~\cite{Chazelle90a} proved that in a pointer machine model of
computation, any data structure for the $d$-dimensional orthogonal
range reporting problem with $O(\log^{O(1)}n + K)$ query time (where
$K$ denotes the output size of the query) requires at least
$\Omega(n (\log n/\log\log n)^{d-1})$ space.  His framework for
proving such lower bounds has been adapted in various subsequent works
(e.g., see \cite{ChazelleR95,AfshaniAL10,AfshaniC21}).

For a given set $\PS$ of $n$ points, Chazelle called
a set of boxes $\BS=\{q_1,\ldots,q_m\}$ \emph{$b$-favorable}
if for each $i,j$ with $i\neq j$,
\begin{enumerate}
    \item[(i)]~$\cardin{\PS\cap q_i} \ge \log^b n$ and

    \item[(ii)]~$\cardin{\PS\cap q_i\cap q_j} \le 1$.
\end{enumerate}
On his way to proving his data structure lower bound, he
constructed a set $\PS$ of $n$ points and a $b$-favorable set $\BS$ of $m$ boxes with
\[ m = \Omega(n(\log n)^{d-b-1}/(\log\log n)^{d-1})
\]
for given any $b,d=O(1)$ and sufficiently large $n$~\cite[Lemma~4.6]{Chazelle90a}.

It is easy to see that condition~(ii) of favorability is equivalent to
the statement that the incidence graph $G(\PS,\BS)$ does not contain
$K_{2,2}$.  On the other hand, condition~(i) implies that the total
number of incidences is $\Omega(m\log^b n)$.

By setting $b$ with $\log^b n = \Theta((\log n/\log\log n)^{d-1})$,
we immediately obtain from Chazelle's construction a set $\PS$ of $n$
points and a set $\BS$ of $n$ boxes in $\Re^d$ such that
$G(\PS,\BS)$ does not contain $K_{2,2}$ and the number of incidences is
$\Omega(n (\log n/\log\log n)^{d-1})$.  (Thus, Chazelle's argument from years earlier not only implied the two-dimensional lower bound from Basit \etal's recent paper, but also answered
their question about extension to higher dimensions~\cite{bcstt-zpsh-21}.)

Similarly, it should be possible to obtain lower bounds for certain other classes of objects by using subsequent
results that follow Chazelle's framework (e.g., simplices~\cite{ChazelleR95}).

\begin{remark}
    It might also be possible to combine Chazelle's framework with
    existing \emph{upper} bounds on orthogonal range reporting in
    pointer machines \cite{Chazelle90a,AfshaniAL10} to obtain upper
    bounds on our incidence problem for boxes.  But query time bounds
    on pointer machines are $\Omega(\log n)$ for $d=2$, so one cannot
    rederive our $O(kn\log n/\log\log n)$ upper bound in two
    dimensions this way.  Also, the dependence on $k$ appears to be
    significantly worse.  It is better to directly modify the ideas
    used in known data structures for orthogonal range reporting
    (which is what we did in \secref{boxes:improved}).

    Independently, Afshani and Cheng's recent work on simplex range reporting lower bounds~\cite{AfshaniCheng23} also exploited the connection with incidence graphs avoiding
    $K_{\alpha,\beta}$ for certain choices of $\alpha$ and $\beta$.
\end{remark}

\end{document}